\documentclass[letterpaper, 10 pt, conference]{ieeeconf}
\IEEEoverridecommandlockouts
\usepackage{cancel}

\usepackage{caption}
\usepackage{subcaption}
\usepackage{graphics} 
\DeclareGraphicsExtensions{ .pdf, .png, .jpg}

\usepackage{mathtools}
\usepackage{physics2}
\usephysicsmodule{ab, diagmat, xmat}

\newcommand{\bmat}[1]{\begin{bmatrix}#1\end{bmatrix}}

\newcommand{\norm}[1]{\ab\|#1\|} 

\usepackage{times}
\usepackage[separate-uncertainty=true]{siunitx}
\usepackage{amssymb}
\usepackage{amsmath}
\usepackage{accents}
\usepackage{amsfonts}
\usepackage[bb=dsserif]{mathalpha}
\usepackage[mathscr]{euscript}
\usepackage{mathrsfs}
\newtheorem{thm}{Theorem}
\newtheorem{lem}{Lemma}
\newtheorem{prop}{Proposition}

\newtheorem{defn}{Definition}

\usepackage{url}

\usepackage{bookmark}

\let\labelindent\relax
\usepackage[inline]{enumitem}

\usepackage{cleveref}
\crefname{assume}{}{}
\Crefname{assume}{}{}
\crefname{figure}{Fig.}{Figs.}
\Crefname{figure}{Fig.}{Figs.}
\crefname{equation}{}{}
\Crefname{equation}{}{}
\crefname{lem}{Lemma}{Lemmas}
\Crefname{lem}{Lemma}{Lemmas}
\crefname{rem}{Remark}{Remarks}
\Crefname{rem}{Remark}{Remarks}
\crefname{thm}{Theorem}{Theorems}
\Crefname{thm}{Theorem}{Theorems}
\crefname{cor}{Corollary}{Corollaries}
\Crefname{cor}{Corollary}{Corollaries}
\crefname{prop}{Proposition}{Propositions}
\Crefname{prop}{Proposition}{Propositions}
\crefname{defn}{Definition}{Definitions}
\Crefname{defn}{Definition}{Definitions}

\makeatletter
\patchcmd{\@IEEEyesnumber}
  {\stepcounter}
  {\refstepcounter}
  {}{}
\patchcmd{\@@IEEEeqnarray}
  {\stepcounter}
  {\refstepcounter}
  {}{}
\patchcmd{\@@IEEEeqnarraycr}
  {\stepcounter{IEEEsubequation}}
  {\refstepcounter{IEEEsubequation}}
  {}{}
\patchcmd{\@@IEEEeqnarraycr}
  {\stepcounter{IEEEsubequation}}
  {\refstepcounter{IEEEsubequation}}
  {}{}
\patchcmd{\@@IEEEeqnarraycr}
  {\stepcounter{IEEEequation}}
  {\refstepcounter{IEEEequation}}
  {}{}
\patchcmd{\@@IEEEeqnarraycr}
  {\stepcounter{IEEEequation}}
  {\refstepcounter{IEEEequation}}
  {}{}
\makeatother

\usepackage{nicematrix}
\NiceMatrixOptions{margin}

\usepackage{booktabs} 
\usepackage{multirow}
\usepackage{makecell}

\usepackage{array}
\newcommand{\PreserveBackslash}[1]{\let\temp=\\#1\let\\=\temp}
\newcolumntype{C}[1]{>{\PreserveBackslash\centering}m{#1}}
\newcolumntype{R}[1]{>{\PreserveBackslash\raggedleft}m{#1}}
\newcolumntype{L}[1]{>{\PreserveBackslash\raggedright}m{#1}}
\usepackage{nicematrix}
\usepackage{arydshln}

\usepackage{tikz}
\usepackage{circuitikz}
\usetikzlibrary{calc,shapes,arrows,arrows.meta,quotes,positioning,angles,quotes, decorations.pathmorphing, patterns, decorations.markings}

\usepackage{todonotes}
\setuptodonotes{inline, inlinewidth=\linewidth}

\newcommand{\dd}[1]{\mathop{}\!\mathrm{d}#1}
\newcommand{\diff}[2]{ \frac{\dd{#1}}{\dd{#2}} }
\newcommand{\T}{\mathsf{T}}

\title{\LARGE \bf
Controller Design for Structured State-space Models via Contraction Theory
}

\author{%
    Muhammad Zakwan$^{\dagger}$, Vaibhav Gupta$^{\dagger}$, Alireza Karimi, Efe C. Balta, Giancarlo Ferrari-Trecate
    \thanks{*This work is funded by the Swiss National Science Foundation (grant no. 200021-204962), NCCR Automation, a National Centre of Competence in Research, funded by the Swiss National Science Foundation (grant no. 51NF40\_225155), and the NECON project (grant no. 200021-219431).}
    \thanks{V. Gupta, G. Ferrari-Trecate, and A. Karimi are with the Laboratoire d’Automatique, EPFL, 1015 Lausanne, Switzerland.}%
    \thanks{M. Zakwan and E. C. Balta are with Control \& Automation Group, Inspire AG, 8005 Zürich, Switzerland \& with Automatic Control Laboratory (IfA), ETH Zürich, 8092 Zürich, Switzerland.}%
    \thanks{$^\dagger$ M. Zakwan and V. Gupta contributed equally to this work.} %
    \thanks{Corresponding author: M. Zakwan, {\small muhammad.zakwan@inspire.ch}}
}

\begin{document}

\maketitle
\thispagestyle{empty}
\pagestyle{empty}

\begin{abstract}
This paper presents an indirect data-driven output feedback controller synthesis for nonlinear systems, leveraging Structured State-space Models (SSMs) as surrogate models. SSMs have emerged as a compelling alternative in modelling time-series data and dynamical systems. They can capture long-term dependencies while maintaining linear computational complexity with respect to the sequence length, in comparison to the quadratic complexity of Transformer-based architectures. The contributions of this work are threefold. We provide the first analysis of controllability and observability of SSMs, which leads to scalable control design via Linear Matrix Inequalities (LMIs) that leverage contraction theory. Moreover, a separation principle for SSMs is established, enabling the independent design of observers and state-feedback controllers while preserving the exponential stability of the closed-loop system. The effectiveness of the proposed framework is demonstrated through a numerical example, showcasing nonlinear system identification and the synthesis of an output feedback controller.
\end{abstract}

\section{Introduction}
\label{sec:intro}
System Identification (SysId) is a foundational pillar of control theory, offering a data-based mathematical representation of an underlying dynamical process and thereby facilitating the analysis, design, and implementation of a wide range of control strategies. While linear SysId techniques have advanced significantly in recent decades due to the availability of well-developed tools, little has been done for nonlinear SysId. The application of linear SysId on nonlinear systems can lead to poor control performance or even instability in practical applications. Consequently, nonlinear SysId has emerged as a pivotal tool in modern control theory.

Nonlinear SysId remains an active field of research, with the choice of optimal nonlinear parametric models still an open question. Model structure often depends on the specific system under consideration, leading to tailored analyses and controller designs that are often limited to particular dynamical systems or applications. A recent survey \cite{schoukens_NonlinearSystemIdentification_2019} summarises various classical frameworks for nonlinear SysId, highlighting their strengths and limitations. Classical approaches include SysId Nonlinear Autoregressive eXogenous (NARX) models \cite{billings2013nonlinear} and kernel-based techniques like Reproducing Kernel Hilbert Spaces (RKHS) or Support Vector Machines (SVMs) \cite{pillonetto2014kernel, schon2011system}. These methods often suffer from challenges, including non-obvious kernel selection, lack of interpretability, scalability, and lack of controller design methods \cite{brunton2022data}. Here, interpretability refers to the ability to understand a model’s behaviour and quantifying properties such as stability, finite or incremental $\mathcal{L}_2$ gain, Lipschitz continuity, and dissipativity. Recently, Machine Learning (ML) models have gained popularity due to increased computational power \cite{chiuso2019system}, yet their intrinsic black-box nature hinders interpretability and controller synthesis. These challenges highlight the need to better integrate system identification with control design.

Recently, Structured State-space Models (SSMs), such as Mamba \cite{gu2024mamba}, have emerged as an alternative to Transformers for sequence modelling. Therefore, they are also natural candidates as surrogate models for nonlinear SysId. A typical SSM consists of a recurrent unit, such as a linear time-invariant dynamical system, surrounded by nonlinear neural-network \emph{scaffoldings} that map the signal into higher dimensions. SSMs have gained significant traction in both the machine learning and control communities \cite{alonso2025state, bonassi2024structured} due to their ability to capture long-term dependencies in time-series data while offering linear complexity with respect to sequence length, in contrast to the quadratic complexity of Transformer-based architectures. Dynamics and control system theory can be employed to analyse and interpret the properties of the recurrent unit, such as stability, to yield an interpretable nonlinear model suitable for controller synthesis. 

In this paper, SSMs are used as surrogate models for indirect controller design synthesis. Specifically, we adopt a variant of SSM, where the recurrent unit is a discrete-time linear time-invariant system called Linear Recurrent Unit (LRU) \cite{orvieto2023resurrecting}, and the scaffoldings are nonlinear NN maps. First, a sufficient structural condition for the controllability and observability of SSMs is derived, which is essential for output feedback controller design. Then, based on contraction theory \cite{lohmiller1998contraction,FB-CTDS} and discrete-time control contraction metrics \cite{manchester2017control}, convex sufficient conditions for the synthesis of a state-feedback controller and a Luenberger-like observer are provided. Interestingly, these conditions take the form of semidefinite programs that can be solved efficiently. The controller and observer pair guarantees global exponential closed-loop stability of the identified nonlinear model. Analogous to linear system theory, a separation principle for controller and observer design is provided, which is based on the auxiliary result on the input-to-state stability of the closed-loop system with a state-feedback controller. 

{\bf Related work: } 
Taking into account the expressivity of Neural Networks (NNs) for nonlinear SysId, a stream of works has considered indirect controller design via Recurrent NNs (RNNs) such as Gated Recurrent Units (GRUs) and Long Short-Term Memory (LSTMs). In these approaches, a surrogate model is identified, and then a controller is designed via Linear Matrix Inequalities (LMIs) \cite{la2024regional}, or the model is used for prediction in Model Predictive Control (MPC) \cite{ravasio2024lmi}. While this approach is promising, there are several caveats. For instance, one has to impose or promote desired properties such as incremental input-to-state ($\delta$-ISS) stability \cite{bonassi2021stability} to employ the model for controller design. Moreover, training LSTMs is more time-consuming compared to SSMs, as the simulation can be performed via the well-established parallel-scan method in the latter, whereas training LSTMs and RNNs requires expensive sequential roll-outs. 

Besides nonlinear RNNs, linear SysId via ML methods has also been explored. In particular, the authors of \cite{di2024simba,di2024stable} have demonstrated that backpropagation and auto differentiation can be leveraged to identify centralised and decentralised linear models that are guaranteed to be stable without compromising expressivity. While these models can be used for controller design, they cannot capture nonlinear intricacies. Another recent work \cite{forgione_DynoNetNeuralNetwork_2021} proposes a computationally friendly framework for nonlinear SysId based on Wiener-Hammerstein models, which uses linear dynamical operators as elementary building blocks in the NN architecture. While these models are highly expressive, easy to train, and exhibit state-of-the-art performance on nonlinear SysId benchmarks, the controller design procedure is not straightforward, therefore limiting their practical usage.

A parallel stream of works focuses on designing or training controllers that guarantee the stability of the closed-loop system \emph{by design}. In this framework, the controller is typically parametrised by an NN such that it ensures closed-loop stability both during and after training. In most cases, stability is guaranteed by compositional properties of dissipative systems. For instance, compositional properties of port-Hamiltonian systems have been considered in \cite{furieri2022distributed}, and dissipative properties such as $\mathcal{L}_2$ gain have been explored in \cite{zakwan2024neural,zakwan2024neural2}. While these methods ensure stability regardless of the choice of parameters, they are heavily dependent on neural Ordinary Differential Equations (ODE) \cite{chen2018neural}, which lose key desirable properties upon discretisation, posing a challenge for real-world implementations. Moreover, training these models along with integrating the ODE takes more time compared to the fast inference of SSMs. Similarly, several results leverage system-level synthesis \cite{furieri2022neural} and internal model control \cite{furieri2024learning} for designing nonlinear controllers for nonlinear systems. Moreover, in \cite{zakwan2024neuralcontrol}, contraction theory \cite{lohmiller1998contraction, FB-CTDS} has been employed for the set-point stabilisation of control-affine systems.  
However, these approaches are model-dependent and computationally expensive to train compared to SSMs. 

The main contributions of the paper are as follows:
\begin{itemize}
    \item Sufficient conditions for the controllability and observability of SSMs with an LRU are established.
    \item LMIs for the synthesis of state feedback and state observer are derived, ensuring the input-to-state stability of the closed-loop system.
    \item Analogous to linear system theory, a separation principle for the state-feedback controller and the observer is presented for the class of SSMs considered.
\end{itemize}
We demonstrate the proposed data-driven output feedback control strategy on a nonlinear DC motor subject to loads, Coulomb friction, and dead zone.

The paper is organised as follows: \Cref{sec:prelims} introduces basic notations, reviews SSMs, and contraction theory. \Cref{ssec:ctrb_obsrb_SSMs} establishes the controllability and observability conditions for SSMs. \Cref{ssec:state_fbk_controller,ssc:observer_design} derive the state feedback controller and state observer, respectively. \Cref{ssec:separation_principle} validates the separation principle for the proposed framework. In \Cref{sec:example}, the framework is applied to a nonlinear DC motor for SysId using SSMs and stabilising controller design. Finally, \Cref{sec:conclusion} concludes the paper and discusses future research directions.

\section{Preliminaries}
\label{sec:prelims}

\paragraph*{Notations}
$\mathbb{R}$ and $\mathbb{C}$ denote the real and complex numbers, respectively. ${M \succ (\succeq) \, N}$ indicates that ${M-N}$ is positive (semi-) definite, $I$ is the identity matrix, and $M^\T$ is the transpose of matrix $M$. The standard Euclidean norm is denoted as  $\norm{\cdot}$. A function ${f: \mathbb{R}^n \mapsto \mathbb{R}^m}$ is ($\mu$, $\nu$)-bi-Lipschitz if
\begin{equation*}
    \mu \norm{x - y} \leq \norm{f(x) - f(y)} \leq \nu \norm{x - y},
    \quad \forall x, y \in \mathbb{R}^n
\end{equation*}
for some  ${0 < \mu \leq \nu}$, and its Jacobian $\mathcal{J}(x)$ satisfies,
\begin{equation*}
     \mu  \leq \underline{\sigma}(\mathcal{J}(x)) \leq \bar{\sigma}(\mathcal{J}(x)) \leq \nu, \quad \forall x \in \mathbb{R}^n
\end{equation*}
where $\underline{\sigma}(\cdot)$ and $\bar{\sigma}(\cdot)$ denote the minimum and maximum singular values, respectively.


\subsection{Structured State-space Models (SSMs)}
SSMs are closely related to RNNs and classical state-space models. However, unlike RNNs, which process sequences iteratively, SSMs utilise global convolution \cite{gu_MambaLinearTimeSequence_2024} or Parallel scan \cite{blelloch_PrefixSumsTheir_1990}, leading to more efficient training and inference.

While several variants of SSMs exist in the literature, in this paper, the primary focus will be on the architecture illustrated in \cref{fig:SSM_architecture}. The fundamental components of an SSM are the `Recurrent Unit' (RU), the nonlinear input lifting ($\mathcal{S}_u$), and the nonlinear output projection ($\mathcal{S}_y$). The input liftings and output projections are commonly referred to as \emph{scaffoldings}. 
\footnote{In general, SSMs refer to deep models composed of multiple layers of the architecture depicted in \cref{fig:SSM_architecture}. However, in this paper, the term `SSM' denotes a single layer rather than a deep model.}
In this paper, it is assumed that $\mathcal{S}_u$ and $\mathcal{S}_y$ are \emph{bi-Lipschitz}. The bi-Lipschitzness property can be verified and computed a posteriori \cite{fazlyab_EfficientAccurateEstimation_2019}, or enforced a priori through structure \cite{araujo_UnifiedAlgebraicPerspective_2022, wang_MonotoneBiLipschitzPolyakLojasiewicz_2024}. This assumption guarantees the controllability and observability of the SSM, as shown in \cref{ssec:ctrb_obsrb_SSMs}.

\begin{figure}[tb]
    \centering
    \begin{tikzpicture}[
        auto, 
        node distance=5.5em,
        >=latex,
        thick,
    ]
    \tikzset{
        NN/.style={
            trapezium, 
            fill=#1!20, draw=#1!33!black, text=black, 
            minimum height=1.5em, minimum width=1.5em, 
            shape border rotate=90,
            },
        block/.style={
            rectangle, 
            fill=#1!20, draw=#1!33!black, text=black, 
            minimum height=2em, minimum width=2em, 
            },
        legend/.style={
            rectangle, 
            fill=#1!20, draw=#1!33!black, text=black, 
            minimum height=1em, minimum width=2em,
            },
        NL/.style={
            circle,
            fill=#1!20, draw=#1!33!black, text=black, 
            minimum height=1.5em,
            inner sep=0pt,
            },
        sum/.style={
            circle,
            fill=#1!20, draw=black, text=black, 
            minimum height=1em,
            inner sep=0pt,
            path picture={
                \draw
                    (path picture bounding box.north west) -- (path picture bounding box.south east)
                    (path picture bounding box.south west) -- (path picture bounding box.north east);
                },
            },
    }

    \node[NN=yellow] (MLP_u) {$\mathcal{S}_u$};
    
    \node[block=cyan, right of=MLP_u] (RU) {\small \makecell{Recurrent\\Unit (RU)}};
    
    \node[NN=yellow, right of=RU, shape border rotate=270] (MLP_y) {$\mathcal{S}_y$};


    \draw[<-] (MLP_u.west) -- +(-2em, 0)
        node[at end, left]{$u$};
    \draw[->] (MLP_u.east) -- (RU.west)
        node[midway, above]{$\bar{u}$};
    \draw[->] (RU.east) -- (MLP_y.west)
        node[midway, above]{$\bar{y}$};    
    \draw[->] (MLP_y.east) -- +(2em, 0)
        node[at end, right]{$y$};

        
        
    
\end{tikzpicture}
    \caption{Typical SSM layer}
    \label{fig:SSM_architecture}
\end{figure}
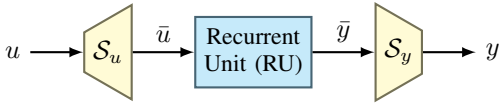

The recurrent unit is typically a discrete-time state-space model. In this paper, we focus on an LRU  \cite{orvieto2023resurrecting}, defined using the following discrete-time state-space equations: 
\begin{subequations}
    \label{eq:S4_RU}
    \begin{align}
        x_{k+1}      &= A x_k + B \bar{u}_k \\
        \bar{y}_k  &= C x_k + D \bar{u}_k 
    \end{align}
\end{subequations}
where
    \({ x \in \mathbb{R}^{n_x} }\), 
    \({ \bar{u} \in \mathbb{R}^{n_{\bar{u}}} }\), and 
    \({ \bar{y} \in \mathbb{R}^{n_{\bar{y}}} }\) 
denote the state, input, and output, respectively.\footnote{$(\bar{u}_k, \bar{y}_k$) denote the LRU input-output pair, while $(u_k,y_k)$ correspond to the SSM \eqref{eq:SSM_nonlinear_model}.} 
The matrices 
    \({ A \in \mathbb{R}^{n_x \times n_x} }\), 
    \({ B \in \mathbb{R}^{n_x \times n_{\bar{u}}} }\),
    \({ C \in \mathbb{R}^{n_{\bar{y}} \times n_x} }\), and
    \({ D \in \mathbb{R}^{n_{\bar{y}} \times n_{\bar{u}}} }\) are trainable parameters. We assume that $(A, B)$ is controllable and $(A, C)$ is observable.
Thus, the nonlinear model for the SSM can be written as:
\begin{subequations}
\label{eq:SSM_nonlinear_model}  
\begin{align}
   x_{k+1} &= A x_k + B \mathcal{S}_u(u_k)  \\
   y_k &= \mathcal{S}_y(Cx_k + D \mathcal{S}_u(u_k) )
\end{align}
\end{subequations}

\subsection{Contraction Analysis}
Contraction theory \cite{lohmiller1998contraction,tsukamoto2021contraction} provides a systematic framework to analyse and ensure stability of discrete-time nonlinear systems along arbitrary, time-varying (feasible) reference trajectories by examining the associated displacement or differential dynamics. Stability analysis and controller synthesis can be jointly addressed through Discrete-time Control Contraction Metrics (DCCMs) \cite{tsukamoto2021contraction}, which guarantee the system’s contraction properties.

To introduce the contraction-based approaches, consider a discrete-time nonlinear control-affine system as follows
\begin{equation}
    \label{eq:contraction:system}
    x_{k+1} = f(x_k) + g(x_k) u_k
\end{equation}
where 
    ${ x_k \in \mathcal{X} \subseteq \mathbb{R}^{n_x} }$ and
    ${ u_k \in \mathcal{U} \subseteq \mathbb{R}^{n_u} }$ denote the system states and the control inputs, respectively.
The corresponding differential dynamics can be given by
\begin{equation}
    \label{eq:contraction:differential_system}
    \delta x_{k+1} = A_k \delta {x_k} + B_k \delta {u_k},
\end{equation}
where 
    ${ A_k = A(x_k) \coloneqq \frac{\partial x_{k+1}}{\partial x_k} }$ and
    ${ B_k = B(x_k) \coloneqq \frac{\partial x_{k+1}}{\partial u_k} }$. 

Consider a state-feedback control law for the differential dynamics \eqref{eq:contraction:differential_system} defined as
\begin{equation}
    \label{equ:contraction:state_fbk}
    \delta {u_k} = K(x_k) \delta {x_k}
\end{equation}
where $K$ is a state-dependent function. 


\begin{defn}
    \label{def:contraction_condition}
    The discrete-time nonlinear system \eqref{eq:contraction:system}, with the associated differential dynamics \eqref{eq:contraction:differential_system} and differential state-feedback controller \eqref{equ:contraction:state_fbk}, is said to be contracting with respect to a uniformly bounded, symmetric, and positive definite metric $M_k = M(x_k) \in \mathbb{R}^{n_x \times n_x}$,
    if for all ${x \in \mathcal{X}}$ and all $\delta x$ in tangent space of $\mathcal{X}$, the following condition holds for some constant contraction rate ${0 < \rho < 1}$:
    \begin{equation}
        \label{eq:contraction_condition}
        (A_k + B_k K_k)^\T M_{k+1} (A_k + B_k K_k) - (1 - \rho) M_k \prec 0
    \end{equation}
    Furthermore, a subset of the state space $\mathcal{X}$ is defined as a `contraction region' if condition \eqref{eq:contraction_condition} holds for every point within that subset.
\end{defn}

\section{Main results}
This section establishes sufficient structural conditions to guarantee the controllability and observability of SSMs.

\subsection{Controllability \& Observability of SSMs}
\label{ssec:ctrb_obsrb_SSMs}
The analysis of controllability and observability of a class of nonlinear systems can be conducted only considering the local controllability and local observability at almost all points, respectively. Readers are referred to \cite{boscain_LocalControllabilityDoes_2023} for more details.

\begin{defn}
A system is locally controllable (or observable) in the neighbourhood of $(x_k, u_k)$ if its differential dynamics around $(x_k, u_k)$ is controllable (or observable). 
\end{defn}

The local controllability of a nonlinear system can be analysed using the controllability of its differential form along the solutions of the system. The differential dynamics for the LRU \eqref{eq:S4_RU} is given by,
\begin{subequations}
    \label{eq:differential_S4_RU}
    \begin{align}
        \delta x_{k+1}     &= A\,\delta x_k + B\,\delta {\bar{u}}_k \\
        \delta {\bar{y}}_k &= C\,\delta x_k + D\,\delta {\bar{u}}_k .
    \end{align}
\end{subequations}

Defining the Jacobian of the scaffolding $\mathcal{S}_u$ and $\mathcal{S}_y$ as $\mathcal{J}^u_k$ and $\mathcal{J}^y_k$, for each $k$ respectively, the differential form of the scaffolding can be written as:
${ \delta {\bar{u}}_k = \mathcal{J}^u_k \,\delta {u}_k }$, 
${ \delta {y}_k = \mathcal{J}^y_k \,\delta {\bar{y}}_k }$.
Hence, the differential form of the SSM \eqref{eq:SSM_nonlinear_model} is,
\begin{subequations}
    \label{eq:differential_S4}
    \begin{align}
        \delta x_{k+1} &= 
            \phantom{\mathcal{J}^y_k} A \,\delta x_k + 
            \phantom{\mathcal{J}^y_k} B \mathcal{J}^u_k \,\delta u_k \\
        \delta y_k     &= 
            \mathcal{J}^y_k C  \,\delta x_k + 
            \mathcal{J}^y_k D \mathcal{J}^u_k \,\delta u_k
    \end{align}
\end{subequations}

\begin{prop}[Local Controllability]
    The SSM model \eqref{eq:SSM_nonlinear_model} with the differential form \eqref{eq:differential_S4} is locally controllable if the recurrent unit is controllable and the input nonlinearity ($\mathcal{S}_u$) is a bi-Lipschitz function.
\end{prop}
\begin{proof}
    Discrete-time controllability Gramian for \eqref{eq:differential_S4} is
    \begin{equation}
        W_d 
            = \sum_{k=0}^{k_1} A^k B
            \mathcal{J}^u_k (\mathcal{J}^u_k)^\T
            B^\T \ab(A^\T)^k
            \label{eq:ctrb_Gramian_S4}
    \end{equation}
    For the SSM model to be locally controllable, the Gramian $W_d$ must be non-singular for some finite $k_1$. Since $(A, B)$ is controllable, $\exists\,\tilde{k}_1$ such that
    \begin{equation*}
        W_d' = \sum_{k=0}^{\tilde{k}_1} A^k B B^\T \ab(A^\T)^k \succ 0 
    \end{equation*}
    Moreover, since the input nonlinearity $\mathcal{S}_u$ is a ($\mu_u$, $\nu_u$)-bi-Lipschitz function, its Jacobian $\mathcal{J}^u$ satisfies
    $
       0 <  \mu_u \leq \underline{\sigma}(\mathcal{J}^u_k) \leq \bar{\sigma}(\mathcal{J}^u_k) \leq \nu_u.
    $
    Utilising these properties along with \eqref{eq:ctrb_Gramian_S4}, it can be observed that for $k_1 = \tilde{k}_1$, one has
    \[
        \nu_u^2 W_d' \succeq W_d \succeq \mu_u^2 W_d' \succ 0.
    \]
    This implies that the Gramian $W_d$ is non-singular.
\end{proof}

\begin{prop}[Local Observability]
    The SSM model \eqref{eq:SSM_nonlinear_model} with the differential form \eqref{eq:differential_S4} is locally observable if the recurrent unit is observable and the output nonlinearity ($\mathcal{S}_y$) is a bi-Lipschitz function.
\end{prop}
\begin{proof}
    The proof parallels that of controllability, employing the discrete-time observability Gramian.
\end{proof}

\subsection{State feedback controller}
\label{ssec:state_fbk_controller}


Consider the discrete-time nonlinear system
\begin{equation}
    \label{eq:state_fbk_plant}
    x_{k+1} = A x_k + B \mathcal{S}_u(u_k),
\end{equation}
where 
    ${ x_k \in \mathbb{R}^{n_x} }$ is the state, 
    ${ A \in \mathbb{R}^{n_x \times n_x} }$ and 
    ${ B \in \mathbb{R}^{n_x \times n_u} }$ are known constant matrices, and 
    ${ \mathcal{S}_u(\cdot): \mathbb{R}^{n_u} \mapsto \mathbb{R}^{n_u} }$ 
is a nonlinear mapping which is ($\mu_u$, $\nu_u$)-bi-Lipschitz, and satisfies $\mathcal{S}_u(0) = 0$. The goal of this section is to design a static state-feedback gain $K$, such that for
${ u_k = K x_k }$
the closed-loop system is exponentially stable for all admissible ($\mu_u$, $\nu_u$)-bi-Lipschitz nonlinearities.

\begin{thm}[State-feedback Controller]
    \label{thm:state_fbk}
    Suppose there exist matrices ${P = P^\T \succ 0}$, $Y \in \mathbb{R}^{n_x \times n_x}$ and $X \in \mathbb{R}^{n_u \times n_x}$ satisfying the following LMI for some ${\sigma > 0}$ and ${\rho_c \in (0,1)}$:
    \begin{equation}
        \label{eq:state_fbk_LMI}
        \begin{bmatrix}
            (1 -\rho_c) P - \sigma B B^\T & A Y + \alpha_u BX & 0           \\
            \star                      & Y^\T + Y - P        & \beta_u X^\T \\
            \star                      & \star              & \sigma I    \\
        \end{bmatrix} \succ 0
    \end{equation}
    where ${ \alpha_u = \frac{\nu_u + \mu_u}{2} }$ and ${ \beta_u = \frac{\nu_u - \mu_u}{2} }$. Then, the nonlinear closed-loop system \eqref{eq:state_fbk_plant} is exponentially stable for bi-Lipschitz functions $\mathcal{S}_u(\cdot)$, and the stabilising controller gain is given by ${ K = X Y^{-1}}$.
\end{thm}
    
\begin{proof}
    Define the variational closed-loop dynamics under the control policy $u_k = K x_k$ and denote the Jacobian of $\mathcal{S}_u(\cdot)$ with respect to input as $\mathcal{J}^u_k$
    \begin{equation}
    \label{eq:closed_loop_diff}
        \delta x_{k+1} = \ab(A + B \mathcal{J}^u_k K) \delta x_k
    \end{equation}
    Consider the candidate Lyapunov function ${V(z) = z^\T P z}$, with ${P = P^\T \succ 0}$. The forward difference satisfies
    \begin{equation*}
    V(\delta x_{k+1}) - V(\delta x_k)
        = (\delta x_k)^\T
            \ab( A_{\text{cl},k}^\T P A_{\text{cl},k} - P )
            (\delta x_k)
    \end{equation*}
    where ${A_{\text{cl},k} = A + B \mathcal{J}^u_k K}$. To ensure exponential decay, it suffices to require that
    \begin{equation}
    \label{eq:LMI_closed_loop}
        A_{\text{cl},k}^\T P A_{\text{cl},k} - (1 - \rho_c) P \prec 0,
    \end{equation}
    for all $\mathcal{J}^u_k$ such that $0 < \mu_u \leq \underline{\sigma}(\mathcal{J}^u_k) \leq \bar{\sigma}(\mathcal{J}^u_k) \leq \nu_u$, which would ensure the contraction rate of $\rho_c$ for the variational dynamics. 
    
    Define the congruence transformation ${ T = \bmat{ I & -A_{\text{cl},k} } }$ which is full row rank, and a free parameter ${ Y \in \mathbb{R}^{n_x \times n_x} }$. Then, \eqref{eq:LMI_closed_loop} can be rewritten as,
    \begin{align}
        T \bmat{
            (1 - \rho_c) P        & A_{\text{cl},k} Y \\
            \star    & Y^\T + Y - P
        } T^\T &\succ 0
        \\
        \iff \bmat{
            (1 - \rho_c) P       & A_{\text{cl},k} Y \\
            \star    & Y^\T + Y - P
        }
        &\succ 0.
        \label{eq:state_fbk_base_LMI}
    \end{align}
    By introducing a change of variables ${X = K Y}$, \eqref{eq:state_fbk_base_LMI} is equal to
    \begin{equation}
        \label{eq:state_fbk_general_LMI}
        \begin{bmatrix}
            (1 - \rho_c) P    & A Y + B \mathcal{J}^u_k X \\
            \star & Y^\T + Y - P
        \end{bmatrix} \succ 0
    \end{equation}
    For some $\Delta$ with $\bar{\sigma}(\Delta) \leq 1$, all $\mathcal{J}_k^u$ can be written as
    \begin{equation*}
        \mathcal{J}_k^u 
            = \underbrace{\ab(\frac{\nu_u + \mu_u}{2})}_{\alpha_u} I
            + \underbrace{\ab(\frac{\nu_u - \mu_u}{2})}_{\beta_u} \Delta
    \end{equation*}
    Then, the inequality \eqref{eq:state_fbk_general_LMI} can be rewritten as,
    \begin{equation}
        \begin{bmatrix}
            (1 - \rho_c) P & A Y + B (\alpha_u I + \beta_u \Delta) X \\
            \star         & Y^\T + Y - P
        \end{bmatrix} \succ 0 
    \end{equation}
    Then, using \cite[Lemma 2]{hindi_ComputingOptimalUncertainty_2002}, the equivalent LMI \eqref{eq:state_fbk_LMI} is obtained. If a feasible solution $(X, Y)$ exists, the stabilising feedback gain can be computed as ${K = X Y^{-1}}$.
\end{proof}

\Cref{thm:state_fbk} provides a convex condition ensuring robust stability of the bi-Lipschitz nonlinear system \eqref{eq:state_fbk_plant}. It captures the uncertainty in the slope of the nonlinearity and guarantees stability for all admissible bi-Lipschitz mappings $\mathcal{S}_u(\cdot)$.

\subsection{Observer design}
\label{ssc:observer_design}

It is well established that control and observer design problems for linear systems enjoy a fundamental and elegant \emph{duality} relation (see, e.g., \cite{hespanha2018linear}). In this section, the result in \cite{manchester2014control}, stating that DCCMs possess an analogous duality relationship to nonlinear observer designs formulated using Riemannian metrics, is leveraged. In particular, we provide a tractable LMI formulation, in contrast to \cite{manchester2014control}, which provides only infinite-dimensional conditions. Furthermore, a novel construction of a Luenberger-like observer for the SSM is presented.

Consider the following Luenberger-like observer:
\begin{subequations}
\label{eq:proposed_observer}
\begin{align}
    \hat{x}_{k+1} &= A\hat{x}_k + B \mathcal{S}_u(u_k) + L ( \hat{y}_k - y_k )
\\
    \hat{y}_k &= \mathcal{S}_y ( C \hat{x}_k + D \mathcal{S}_u(u_k) )
\end{align}
\end{subequations}
where $L \in \mathbb{R}^{n_x \times n_y}$ is the observer gain to be computed. 

\begin{thm}[State Observer]
    \label{thm:observer_design}
    Suppose there exist matrices ${Q = Q^\T \succ 0}$, ${U \in \mathbb{R}^{n_x \times n_x}}$ and ${V \in \mathbb{R}^{n_x \times n_y}}$ satisfying the following LMI for some $\eta > 0$ and $\rho_o \in (0,1)$: 
    \begin{equation} 
        \label{eq:state_obs_LMI}
        \begin{bmatrix}
            (1 -\rho_o) Q - \eta C^\T C & (UA + \alpha_y VC)^\T & 0         \\
            \star                       & U + U^\T - Q          & \beta_y V \\
            \star                       & \star                 & \eta I    \\
        \end{bmatrix} \succ 0 .
    \end{equation}
    where ${ \alpha_y = \frac{\nu_y + \mu_y}{2} }$ and ${ \beta_y = \frac{\nu_y - \mu_y}{2} }$. Then, the nonlinear observer \eqref{eq:proposed_observer} is exponentially stable for all bi-Lipschitz functions $\mathcal{S}_y(\cdot)$, and the observer gain is given by ${L = U^{-1} V}$.
\end{thm}

\begin{proof}
    The proof is done in two steps. 
    First, it is shown that the proposed observer is considered `correct' according to the definition of correctness provided in \cite{manchester2018contracting}. Specifically, when the proposed observer is initialised with ${\hat{x}_0 = x_0}$, the observer matches the true system, i.e., ${\hat{x}_k = x_k}$ for all $k \geq 0$. If ${\hat{x}_0 = x_0}$, from \eqref{eq:proposed_observer},
    ${
        \hat{x}_1 
            = A\hat{x}_0 + B \mathcal{S}_u(u_0) + \cancelto{0}{L ( \hat{y}_0 - y_0 )}
            = x_1
    }$.
    Using induction, it can be easily proved that the proposed observer is `correct'.

    Second, if the LMI condition outlined in \eqref{eq:state_obs_LMI} is satisfied, then the proposed observer exhibits global exponential stability. This indicates that the error between $x_k$ and $\hat{x}_k$ converges to zero at an exponential rate, as stated in \cite{manchester2018contracting}. 
    Similar to the proof for the state-feedback controller, the variational dynamics of the nonlinear observer \eqref{eq:proposed_observer}, which describes the evolution of an infinitesimal displacement $\delta \hat{x}_k$ between two neighbouring observer trajectories under identical input, is considered with the change of variables ${V = U L}$. Furthermore, the observer gain can be recovered as ${L = U^{-1} V}$. 
    %
\end{proof}


\subsection{Separation principle for the SSMs}
\label{ssec:separation_principle}

If a linear system is both controllable (or stabilizable) and observable (or detectable), the controller and observer design can be done independently. This extremely useful property, known as the `separation principle', does not hold for general non-linear systems. In \cite{manchester2014output}, it is demonstrated that the separation principle holds for a continuous-time nonlinear system if it is universally stabilizable and detectable. 
Numerous studies have also addressed the separation principle for nonlinear systems, such as those in \cite{atassi2002separation, shiriaev2008separation, martinez2017separation}. However, these approaches often impose structural constraints on the nonlinear model or rely on high-gain observers.
The current section establishes the separation principle for the discrete-time SSM \eqref{eq:SSM_nonlinear_model}. We start by introducing a preliminary result. 

\begin{lem}[Discrete-time contraction with disturbance]
    \label{thm:dt_contraction_disturbance}
    Consider the discrete-time system
    \begin{equation}
        x_{k+1} = A x_k + B \mathcal{S}_u(u_k) + w_k
    \label{eq:dt_sys}
    \end{equation}
    where $w_k$ is a disturbance input. Let the state-feedback $u_k= K x_k$ be the state-feedback controller designed in \Cref{thm:state_fbk}. Denote the closed-loop map by ${\Phi(x) \coloneqq A x_k +  B \mathcal{S}_u(K x_k)}$ and assume there exists a smooth metric ${M = \Theta^\T \Theta \succ 0}$ and a constant ${\rho \in (0,1)}$ such that for all $x_k$
    \begin{equation} 
        \label{eq:dt_diff_contraction}
        \Theta \frac{\partial \Phi}{\partial x}(x_k) \Theta^{-1} = F(x_k)
    \end{equation}
    where ${\norm{F(x_k)} \leq \rho}$. Furthermore, assume that ${\norm{\Theta} \leq c_\Theta}$ for all $x_k$. Let ${d_k \coloneqq d(x_k, x_k^\star)}$ be the Riemannian distance from $x_k$ to $x_k^\star$ with respect to the metric $M$. Then,
    \begin{equation}
        \label{eq:dt_one_step}
        d_{k+1} \leq \rho d_k + c_\Theta \norm{w_k}, \quad \forall k
    \end{equation}
\end{lem}

\begin{proof}
    Let ${\gamma_k: [0,1] \mapsto \mathbb{R}^{n_x}}$ be a unit-speed geodesic joining ${x_k^\star = \gamma_k(0)}$ to ${x_k = \gamma_k(1)}$ in the metric $M$. Set ${\delta x_k(s) \coloneqq \frac{\partial\gamma_k}{\partial s}}$ satisfying 
    \[
        \delta x_{k+1} = (A + B \mathcal{J}^u_k K) \delta x_k + w_k
    \]
    Define the differential coordinates ${\delta z_k(s) \coloneqq \Theta \delta x_k(s)}$ so that $\norm{\delta z_k(s)}$ is the differential line element in the Riemannian metric
    and the Riemannian distance is
    \begin{equation*}
        \label{eq:dt_distance_def}
        d_k = d(x_k, x_k^\star) = \min_{\gamma} \int_{0}^1 \norm{\delta z_k(s)} \dd{s}
    \end{equation*}
    In metric coordinates,
    \begin{equation*}
        \delta z_{k+1}
            = \underbrace{
                \Theta \frac{\partial \Phi}{\partial x}(x_k) \Theta^{-1}
            }_{F(x)}
            \delta z_k(s) + \Theta w_k.
    \end{equation*}
    By the contraction hypothesis, $\norm{F(x)} \leq \rho$. Hence,    
    \begin{align*}
        \norm{\delta z_{k+1}}
            \leq \rho \norm{\delta z_k} + \norm{\Theta} \norm{w_k}
            \leq \rho \norm{\delta z_k} + c_\Theta \norm{w_k} .
    \end{align*}
    Integrating over $s \in[0,1]$ 
    \begin{align*}
        d_{k+1} 
            &= \int_0^1 \norm{\delta z_{k+1}(s)} ds
            \leq \rho \int_0^1 \norm{\delta z_k(s)} ds + c_\Theta \norm{w_k} \\
            &= \rho d_k + c_\Theta \norm{w_k}
    \end{align*}
    gives the one-step bound on the Riemannian distance between $x_k$ and $x_k^\star$.
\end{proof}


\begin{thm}[Separation principle]
    Consider a discrete-time SSM \eqref{eq:SSM_nonlinear_model} that is both observable and controllable. The closed-loop system, which incorporates the state observer as defined in \eqref{eq:proposed_observer} and the state-feedback controller as defined in \Cref{thm:state_fbk} using the estimated state $\hat{x}_k$, exhibits exponential stability.
\end{thm}

\begin{proof}
    The closed-loop can be written as 
    \begin{align*}
        x_{k+1} 
            &= A x_k + B \mathcal{S}_u(u_k)
    \\
        \hat{x}_{k+1} 
            &= A \hat{x}_k + B \mathcal{S}_u(u_k) + L ( \hat{y}_k - y_k )
    \\
        y_k &= \mathcal{S}_y(C x_k + D \mathcal{S}_u(u_k))
    \\
        \hat{y}_k 
            &= \mathcal{S}_y(C \hat{x}_k + D \mathcal{S}_u(u_k))
    \\
        u_k &= K \hat{x}_k
    \end{align*}
    where $K$ and $L$ denote the state-feedback and observer gains, as defined in \Cref{thm:state_fbk} and \Cref{thm:observer_design}, respectively. 

    From \Cref{thm:observer_design}, the state estimates $\hat{x}_k$ converge exponentially fast to the true $x_k$. Furthermore, the smoothness of $K$ implies that ${(K\hat {x}_k - K x_k)}$ is bounded and converges to zero asymptotically and since $\mathcal{S}_u(0) = 0$, and $\mathcal{J}_u$ is uniformly bounded, ${\mathcal{S}_u(K \hat{x}_k) - \mathcal{S}_u(K x_k) \rightarrow 0}$, exponentially. So, there exist $\Theta$ and some ${\rho_o \in (0,1)}$, ${c > 0}$ such that,
    \[
        \norm{
            \Theta B \ab(
                \mathcal{S}_u(K \hat{x}_k) 
                - \mathcal{S}_u(K x_k) 
            )
        } \leq c \rho_o^k
    \]
    Now using \Cref{thm:dt_contraction_disturbance} with $
        w_k = B \ab(
            \mathcal{S}_u(K \hat{x}_k) 
            - \mathcal{S}_u(K x_k) 
        )
    $
    and some constant $\rho_c \in (0, 1)$, it follows that
    \[
        d_k \leq 
            \rho_c^k d_0
            + c \sum_{i=0}^{k-1}{\rho_c^{k-1-i} \rho_o^i} 
    \]
    implying that $d_k \to 0$ at an exponential rate, and the uniform boundedness of the metric $\Theta$ guarantees that $x_k$ converges to the desired trajectory $x_k^*$ at an exponential rate.
\end{proof}


\section{Numerical Example}
\label{sec:example}

We consider the nonlinear DC motor model proposed in \cite{kara2004nonlinear}, which captures the key nonlinearities, such as the input dead-zone and nonlinear friction in the drive train. This model is used as the benchmark plant for data collection, SysId, controller synthesis, and observer design. It comprises electrical and mechanical subsystems subject to nonlinearities. 

\paragraph*{Electrical Subsystem} 
The armature voltage dynamics of the DC motor are governed by
\begin{equation*}
    \label{eq:armature_voltage}
    v_a(t) = R_a i_a(t) + L_a \diff{i_a(t)}{t} + e_a(t) 
\end{equation*}
where 
    $v_a(t)$ is the motor armature voltage, 
    $R_a$ and $L_a$ are the armature coil resistance and inductance, 
    $i_a(t)$ is the armature current, and 
    ${ e_a(t) = K_m \omega_m(t) }$ is the back electromotive force (EMF) with $K_m$ being the motor torque constant and $\omega_m(t)$ the motor angular velocity. The motor torque is linearly related to the current as $T_m(t) = K_m i_a(t)$.


\paragraph*{Mechanical Subsystem} 
The mechanical part of the system is modelled as a two-mass drive with elastic coupling between the motor and the load, given by
\begin{align*}
    J_m \dot{\omega}_m(t) &= T_m(t) - T_s(t) - B_m \omega_m(t) - T_f(\omega_m) \\
    J_L \dot{\omega}_L(t) &= T_s(t) - B_L \omega_L(t) - T_d(t) - T_f(\omega_L)
\end{align*}
where $J_m$ and $J_L$ are the moment of inertia of the motor and load, $B_m$ and $B_L$ are viscous friction coefficients, and $T_d(t)$ is the external disturbance torque. The coupling torque $T_s(t)$ between motor and load is modelled as
\begin{equation*}
    T_s(t) 
        = k_s \ab( \theta_m(t) - \theta_L(t) ) 
        + B_s \ab( \omega_m(t) - \omega_L(t) )
\end{equation*}
where $k_s$ and $B_s$ are the shaft stiffness and damping coefficients. Moreover, $\dot{\theta}_m = \omega_m$, and $\dot{\theta}_L = \omega_L$, respectively. 


\paragraph*{Nonlinear Friction and Dead-zone Effects}
The friction torque is modelled using the Coulomb characteristic,
\begin{equation*}
    \label{eq:friction}
    T_f(\omega) = a_0 \operatorname{sgn}(b_0 \omega),
\end{equation*}
where $a_0$ and $b_0$ are friction parameters. The dead-zone nonlinearity in the input voltage is modelled as
\begin{equation*}
    u_{\text{eff}}(t) =
    \begin{cases}
        0,                                                      & \ab|v_a(t)| <    v_{\text{dz}} \\
        v_a(t) - \operatorname{sgn}(v_a(t)) \, v_{\text{dz}},   & \ab|v_a(t)| \geq v_{\text{dz}}
    \end{cases}
    \label{eq:deadzone}
\end{equation*}
where $v_{\text{dz}}$ denotes the dead-zone threshold voltage. The $\operatorname{sgn}(\cdot)$ is the sign function.
The parameters used in this work are summarised in \Cref{tab:params}.


\begin{table}[tb]
    \centering
    \caption{DC Motor Parameters}
    \label{tab:params}
    \begin{tabular}{@{} l c S[table-alignment-mode = marker, table-format = 2.2e1] l @{}}
        \toprule
            \thead{Parameter}
                & \thead{Symbol}    
                & \multicolumn{2}{c}{\thead{Value}} \\
        \midrule
            Armature resistance & $R_a$             & 3.3       & \unit{\ohm}               \\
            Armature inductance & $L_a$             & 2.75      & \unit{\milli\henry}       \\
            Motor constant      & $K_m$             & 3.24e-2   & \unit{\N\m\per\A}         \\
            Motor inertia       & $J_m$             & 1.16e-4   & \unit{\kg\m\squared}      \\
            Load inertia        & $J_L$             & 4.0e-4    & \unit{\kg\m\squared}      \\
            Shaft stiffness     & $k_s$             & 1.35      & \unit{\N\m\per\radian}    \\
            Viscous friction    & $B_m$, $B_L$      & 1.0e-4    & \unit{\N\m\s\per\radian}  \\
            Dead-zone threshold & $v_{\text{dz}}$   & 0.4       & \unit{\volt}              \\
        \bottomrule
    \end{tabular}
\end{table}


For training the SSM, several trajectories are first gathered by exciting the model with a standard PRBS signal. The measured output is corrupted by zero-mean Gaussian noise with variance $0.02$.
The architecture depicted in \cref{fig:SSM_architecture} is used with a NN consisting of a single hidden layer with $32$ neurons and leaky ReLU as the activation function, denoted by $\mathcal{S}_u$, and a linear layer, denoted by $\mathcal{S}_y$. The slope of the leaky ReLU is constrained between $0.01$ and $1$. Note that the bi-Lipschitz bound can be computed after training using the spectral norms of the weight matrices and the slope of the activation function. The system order for the linear recurrent unit is chosen as ${n_x = 8}$, ${n_{\bar{u}} = 1}$, and ${n_{\bar{y}} = 1}$. \Cref{fig:validation_traj} shows a validation trajectory generated using the trained SSM. 

\begin{figure}[tb]
    \centering
    \includegraphics[width=0.9\linewidth]{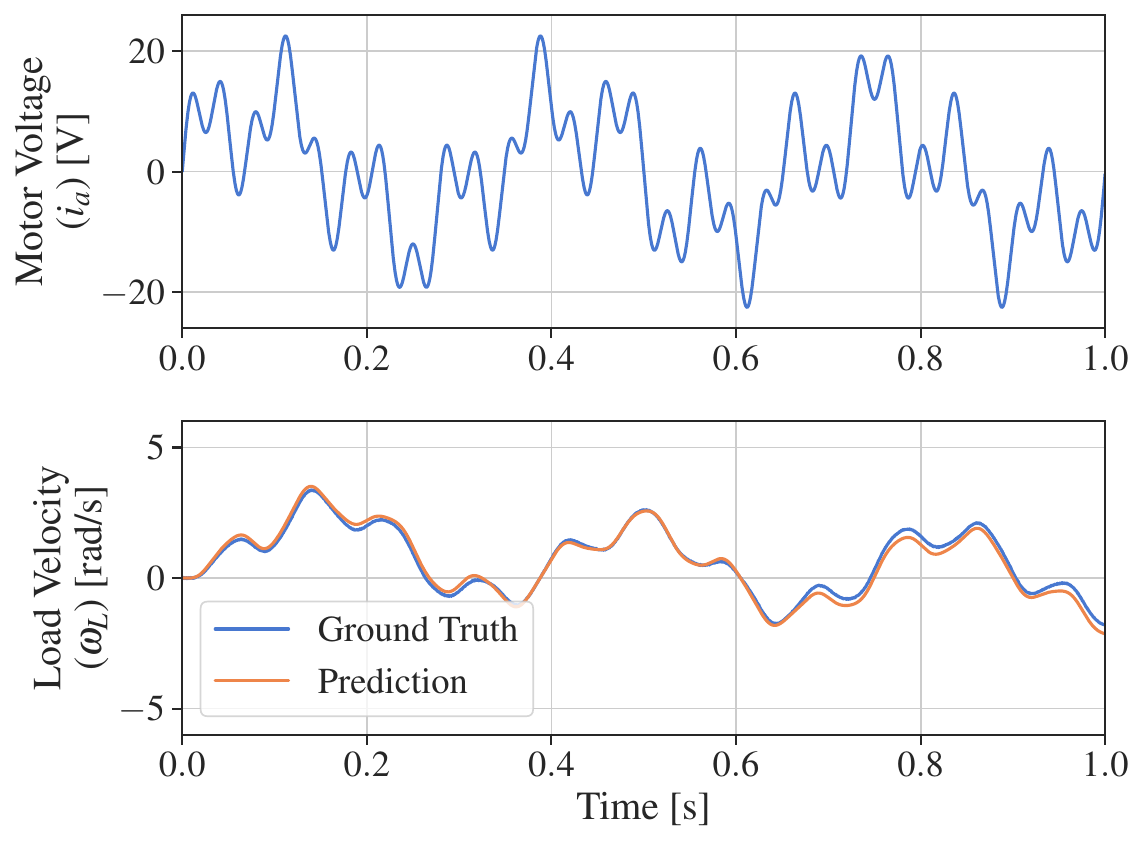}
    \caption{A sample validation trajectory after the training}
    \label{fig:validation_traj}
\end{figure}

For the controller synthesis, the Python Control package \cite{python-control2021} is used to solve the LMIs presented in \cref{thm:state_fbk,thm:observer_design}. The controller and observer pair are implemented on the original nonlinear mathematical model for ten different initial conditions sampled from a uniform distribution between \SIlist{-4; 4}{\radian\per\second}. The results are presented in \cref{fig:response_plant}. The controller and observer demonstrate robustness and stability on the nonlinear mathematical model.
\begin{figure}
    \centering
    \includegraphics[width=0.9\linewidth]{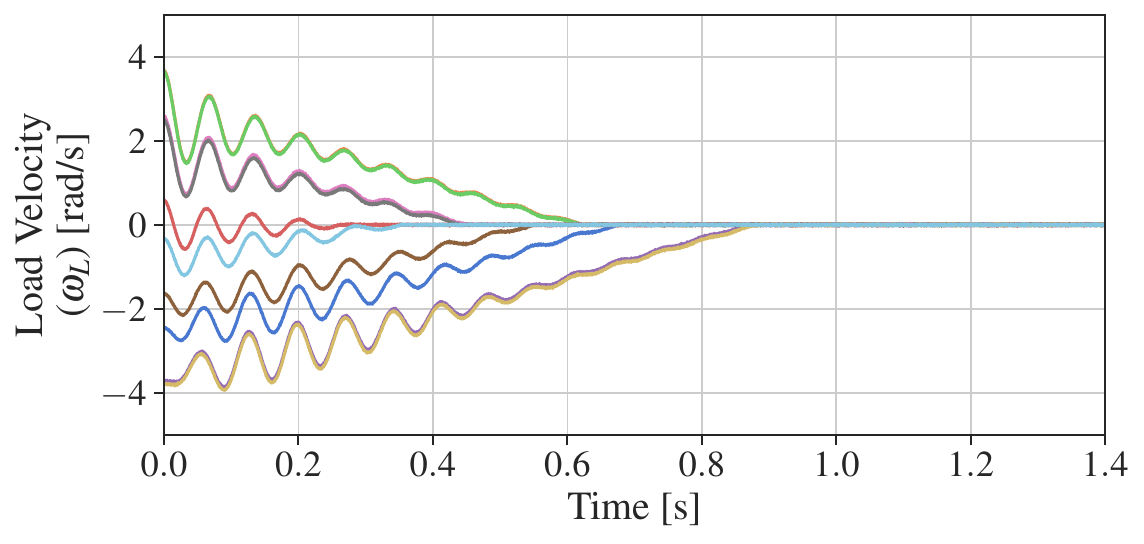}
    \caption{Controller response on the nonlinear mathematical model for different initial conditions}
    \label{fig:response_plant}
\end{figure}

\section{Conclusions}
\label{sec:conclusion}

This paper presents an indirect data-driven controller synthesis for nonlinear systems. First, a nonlinear SysId using an SSM is performed on the system's input-output data, followed by controller synthesis for the learned model. Sufficient controllability and observability conditions for SSMs are established, showing that bi-Lipschitzness of both input lifting and output projection blocks is a sufficient requirement. Key results include
\begin{enumerate*}[label = (\roman*), before=\unskip{: }, itemjoin={{, }}, itemjoin*={{, and }}]
    \item an LMI-based state-feedback controller ensuring exponential stability
    \item a state-observer design guaranteeing asymptotic convergence
    \item a discrete-time separation principle using contraction theory.
\end{enumerate*}
Future work will focus on integrating the internal model principle to enhance robustness and performance.

\bibliographystyle{ieeetr}
\bibliography{references}

\end{document}